\documentclass[aps,pra,
twocolumn,
superscriptaddress,
10pt]{revtex4}
\usepackage{amsmath,amsthm,amssymb,graphicx}
\usepackage{bm,mathrsfs,dsfont,nicefrac}
\usepackage{braket,todonotes,comment}
\usepackage{xcolor}
\usepackage{float}
\usepackage{lscape}
\theoremstyle{remark}

\usepackage{algorithm, algpseudocode}
\usepackage{algorithmicx}
\usepackage{stmaryrd} 
\usepackage{youngtab}
\usepackage{ytableau}

\newcommand{\dket}[1]{|#1\rangle\!\rangle} 
\newcommand{\dbra}[1]{\langle\!\langle #1}

\algdef{SE}[VARIABLES]{Variables}{EndVariables}
{\algorithmicvariables}
{\algorithmicend\ \algorithmicvariables}
\algnewcommand{\algorithmicvariables}{\textbf{Global variables}}

\newtheorem{theorem}{Theorem}

\newtheorem{lemma}[theorem]{Lemma}

\newtheorem{remark}{Remark}
\newtheorem{example}{Example}

\algnewcommand{\algorithmicgoto}{\textbf{go to}}%
\algnewcommand{\Goto}[1]{\algorithmicgoto~\ref{#1}}%

\begin{document}
	
	\preprint{APS/123-QED}
	
	\title{
		Exact quantum noise deconvolution with partial knowledge of noise }
	\author{Nahid Ahmadvand}
	\author{Laleh Memarzadeh}
	\email{memarzadeh@sharif.edu}
	\affiliation{%
		Department of Physics, Sharif University of Technology, Tehran 11155-9161, Iran.
	}%
	\date{\today}
	
	\begin{abstract}
		 We introduce a new quantum noise deconvolution technique that requires neither complete knowledge of the noise nor partial noise tomography, and is applicable to any invertible quantum noise channel. In this new method, we construct a set of observables with completely correctable expectation values despite our incomplete knowledge of noise. This task is achieved just by classical post-processing without extra quantum resources. We show that the number of parameters in the subset of observables with correctable expectation values is the same for all unitary quantum channels. For mixed unitary channels and the assumption that the probability distribution of unitary errors is unknown, we instruct the construction of the set of observables with correctable expectation values. For a particular case where the mixed unitary channel is made of just two random unitary Kraus operators acting on $d$-dimensional Hilbert-space, we show that the observable with correctable expectation value belongs to a set with at least $d$ parameters. We extend our method by considering observables for which the partial recovery of the expectation value is possible, at the cost of having partial knowledge about the noise-free initial state.

	\end{abstract}
	
	\maketitle
	
	\section{Introduction}
	Quantum noise management is a fundamental challenge to the practical realization of quantum technology. Due to the nature of quantum systems, their complete isolation from environmental interaction is impossible. These unavoidable interactions cause errors on quantum systems, which limit the performance of quantum tasks, including fault-tolerant quantum computation \cite{Knill1998, Aharonov2008}, quantum communication \cite{Holevo1973, Holevo1998}, quantum network development \cite{Kimble2004, Razavi2018, Wehner2018} and quantum sensing \cite{Giovannetti2006, Degen2017}. Therefore, developing strategies to combat noise effects is essential for practical applications of quantum information science. 
	
	To remove noise effects, it is essential to describe the evolution of a quantum system, including evolution due to noise, by a mathematical map that is independent of the state it acts upon. Such a map is proved to be a completely-positive trace-preserving (CPTP) map \cite{breuer2010, rivas2012}.
	Therefore, to invert the noise effect, one may suggest finding the inverse of this map (if it exists). But the inverse may not be a CPTP map. That implies that even if it is possible to find the mathematical inverse of the noise map, implementing this inverse is not physically possible unless the inverse is a CPTP map. Therefore, more intelligent approaches are required to remove the noise effects.
	
	The conventional method for addressing noise effects has been quantum error correction (QEC)  protocols. Theoretically, QEC enables detection and correction of errors. However, it adds hardware overhead that makes its implementation challenging, if not impossible, with today's technology \cite{Knill2005,Cao22, Preskill2018}. Regarding this obstacle, alternative strategies are developed, like reversing noise \cite{karimipour2020,Shahbeigi2021}. This method requires the implementation of the reverse of the noise, which is not always a CPTP map. Hence, complete recovery is not possible unless the noise is a unitary conjugation. Another technique is quantum feedback control to retrieve information \cite{Gregoratti2003, Memarzadeh2011, Memarzadeh2011_Cafaro}.  But this technique requires access to the information dissipated to the environment. One of the most well-established techniques is quantum error mitigation (QEM) \cite{Temme2017, Suguru2018, Sam2019, Cai2023}. QEM reduces the impact of noise on the expectation value of observables without the need for extra hardware resources. There are a variety of successful QEM methods, including zero-noise extrapolation \cite{Li2017, Temme2017}, probabilistic error correction \cite{Temme2017}, measurement error mitigation \cite{Berg2022}, with primary application in quantum computation.

	Recently, there have been successful noise deconvolution methods \cite{mangini2022, roncallo2023} for removing generic noise effects on the expectation value of observables suitable for quantum communication and quantum tomography \cite{Siddhu2019, DAriano2000, DAriano2003, Cramer2010}. In \cite{mangini2022} it has been demonstrated that for a qubit system, noise-free expectation value can be recovered, without actual inversion of noise in the lab. Therefore, in this technique, we are not concerned that the inverted noise map is not physical. The power of this method relies on classical post-processing.  The work in \cite{mangini2022} reveals that if the noise is known, the expectation value of an observable can be obtained by measuring a modified observable on the noisy state. The method \cite{mangini2022} instructs deciding the modified observable and, by post-processing, cancels the noise impact. These results have been extended to multi-qudit channels \cite{roncallo2023}.
	
	The technique proposed in \cite{mangini2022} requires full knowledge of noise. In \cite{roncallo2023}, it has been shown that if the noise is unknown by initial state preparation and partial process tomography, it is possible to recover the noise-free expectation value. But the method is not efficient for a general noise model.  In this work, we introduce a new noise deconvolution technique that does not require full knowledge of noise and partial process tomography. This technique significantly improves our practical capability in eliminating noise effects.
	
	In our technique, we give instructions for determining a set of observables with completely recoverable expectation values and explain how to achieve this goal. Similar to the method introduced in \cite{roncallo2023, mangini2022}, our approach is generic, but unlike them, with our technique, complete recovery of the expectation value does neither require full knowledge of the noise nor initial state preparation. Our method is not restricted to a particular class of quantum noise models or system dimensions. We also discuss how one can go beyond the set of observables with completely correctable expectation values.  We explain how to reduce the noise effect on the expectation value of these additional observables. 
	
	The structure of the paper is as follows: In \S\ref{sec:Background} we set the notation and discuss the preliminaries. 
	We introduce our proposed method for noise deconvolution in \S\ref{sec:GenFormalism}. 
	We discuss full recovery of the expectation value of observables in \S\ref{sec:ComRec} with details on the general noise model and then focus on mixed unitary channels. 
	In \S\ref{sec:ParRec} we address observables  beyond the perfectly recoverable set.  Discussion of results and conclusion is presented in \S\ref{sec:Dis}.
	
	\section{background}
	\label{sec:Background}
	In this section, we first set our notation. Then we review the background on quantum channels and the noise deconvolution technique.
	
	In our analysis, the initial and final Hilbert spaces are $d$-dimensional and denoted by $\mathscr{H}_d$. The set of linear operators on ${\mathscr{H}}$ is denoted by $\mathcal{L}(\mathscr{H})$, and the set of trace-one positive operators, called density operators, is denoted by $\mathcal{D}(\mathscr{H}_d)$.
	
	\subsection{Quantum noise}
	In this subsection, we review the main definitions and properties of quantum channels that we require for quantum noise deconvolution. After the primary definition of a quantum channel, we discuss its superoperator representation of the channel, its adjoint, and inverse. 
	
	Quantum noise is described by a completely positive trace preserving (CPTP) map or a quantum channel. A quantum channel $\Phi:\mathcal{L}(\mathscr{H}_d)\to\mathcal{L}(\mathscr{H}_d)$ is represented by its Kraus operators $\{A_k\}$ as follows \cite{choi1975, kraus1983}:
	\begin{equation}
		\label{eq:KrausRep}
		\Phi(\bullet) = \sum_{k=1}^{K} A_k \bullet A_k^\dagger ,\;\;\; \sum_{k}A_k^\dagger A_k=\mathbb{I}_d,
	\end{equation}
	where $\bullet$ belongs to $\mathcal{L}(\mathscr{H}_d)$, and $\mathbb{I}_d$ is $d$-dimensional identity operator in $\mathcal{L}(\mathscr{H}_d)$.
	
	There is an equivalence relation between quantum channels. Channel $\Phi$ is equivalent to channel $\mathcal{E}$, if there exist unitary conjugations  $\mathcal{U}$ and $\mathcal{V}$, such that
	\begin{equation}
		{\Phi}=\mathcal{U}\circ\mathcal{E}\circ\mathcal{V}.
	\end{equation}
	The unitary conjugation by unitary operator $U$ is defined by $\mathcal{U}[\bullet]:=U\bullet U^\dagger$. For our purpose, it is more convenient to vectorize elements of $\mathcal{L}(\mathscr{H}_d)$ and represent the quantum channel as a linear transformation on a vector space $\mathscr{V}_{d^2}$ of dimension $d^2$. More specifically, by vectorization we transform any operator $M\in\mathcal{L}(\mathscr{H}_d)$ into a vector $\dket{M}\in\mathscr{H}_{d^2}$, with $M_{i,j}=\dket{M}_{id+j}$, for $i,j\in\{0,1,\cdots d-1\}$.
	By this notation, we can relate the inner product in 
	$\mathscr{H}_{d^2}$ 
	to
	Hilbert-Schmidt inner product in $\mathcal{L}(\mathscr{H}_d)$:
	\begin{equation}
		\forall X,Y\in\mathcal{L}(\mathscr{H}_d): \;\;\dbra{X}\dket{Y}:=\operatorname{Tr}(X^\dagger Y).
	\end{equation}
	
	By vectorization, the superoperator representation of  the quantum channel $\Phi$ in Eq.~(\ref{eq:KrausRep}) is given by
	\begin{equation}
		\label{eq:MatrixRep}
		\Gamma_{\Phi}=\sum_k A_k\otimes A_k^*,
	\end{equation}
	where $A_k^*$ is complex conjugation of $A_k$ and Eq.~(\ref{eq:KrausRep}) is transformed to
	\begin{equation}
		\dket{\Phi({\bullet})}=\Gamma_{\Phi}\dket{\bullet}.
	\end{equation}
	Another representation of the channel is given by its Choi-Jamio\l kowski matrix, defined as follows: 
	\begin{equation}
		\label{eq:ChoiJamiol}
		C_{\Phi}=(\Phi\otimes\operatorname{id})\ket{\Omega}\bra{\Omega},
	\end{equation}
	where $\operatorname{id}$ is the identity map and $\ket{\Omega}=\frac{1}{\sqrt{d}}\sum_{i=0}^{d-1}\ket{i,i}$ is the maximally entangled state in $d$-dimensional Hilbert space $\mathscr{H}_d$. With an involution $\rightthreetimes$, that is called reshuffling,  $\bra{i,j}C^{\rightthreetimes}_{\Phi}\ket{k,l}=\bra{i,k}C_{\Phi}\ket{j,l}$, the 
	Choi-Jamio\l kowski matrix is transformed to the superoperator representation of the channel:
	\begin{equation}
		\label{eq:ChoiMatrixForm}
		C^{\rightthreetimes}_{\Phi} =d\Gamma_{\Phi}.
	\end{equation}
	This enables us to construct the superoperator representation of the channel directly from its Choi-Jamio\l kowski representation. 
	It is worth recalling that the adjoint of the map $\Phi$ in Eq.~(\ref{eq:KrausRep}) is given by
	\begin{equation}
		\label{eq:AdjVecRep}
		\Phi_{\rm adj} (\bullet) := \hat{\Phi}(\bullet) = \sum_{k=1}^{K} A_k^\dagger \bullet A_k. 
	\end{equation}
	Hence, the adjoint of a completely positive (CP) map admits a Kraus representation, which proves that $\widehat{\Phi}$ is a CP map. The superoperator representation of $\Phi_{\rm adj}$ is given by
	\begin{equation}
		\label{eq:AdjKrausVecRep}
		\Gamma_{\widehat{\Phi}}=\sum_{k}A_k^\dagger\otimes A_k^\top,
	\end{equation}
	where $\top$ denotes transposition. By comparing Eq.~(\ref{eq:MatrixRep}) and Eq.~(\ref{eq:AdjKrausVecRep}) it becomes clear that
	\begin{equation}
		\label{eq:MatrixAdj}
		\Gamma_{\widehat{\Phi}}=\Gamma_\Phi^{\dagger}.
	\end{equation}
	Another useful map for our analysis is the inverse of a CPTP map $\Phi$, which we denote by
	$\Phi^{-1}$. By definition if $\Phi$ is invertible
	\begin{equation}
		\label{eq:InverseMap}
		\Phi\circ
		\Phi^{-1}=\Phi^{-1}\circ\Phi=\operatorname{id}.
	\end{equation}
	The inverse of a CPTP map is not necessarily a CPTP map, but as long as the determinant of $\Phi$ (which is equal to the determinant of the superoperator representation of $\Phi$, denoted by $\Gamma_{\Phi}$) is non-zero, $\Phi^{-1}$ exists. By using the superoperator representation of $\Phi$ and $\Phi^{-1}$ in Eq.~(\ref{eq:InverseMap}) we have:
	\begin{equation}
		\Gamma_{\Phi^{-1}}=\Gamma_{\Phi}^{-1}.
	\end{equation}
	By considering the adjoint of the composed map in Eq.~(\ref{eq:InverseMap}), it is seen that the adjoint of  the inverse of $\Phi$ is the inverse of the adjoint of $\Phi$:
	\begin{equation}
		\label{eq:AdjInv}
		\widehat{\Phi^{-1}}=\widehat{\Phi}^{-1}.
	\end{equation}
	Hence, we can switch the order of taking an adjoint and inverse of the map.  
	\subsection{Quantum noise deconvolution}
	This subsection is devoted to quantum noise deconvolution. We recall the main idea of the quantum noise deconvolution technique discussed in \cite{mangini2022, roncallo2023}. Also, we review the necessary basis of quantum process tomography. 
	
	The noise deconvolution technique is based on the following equality:
	\begin{align}
		\label{eq:mainND}
		\langle A\rangle_{\rho}&=\operatorname{Tr}(A\rho)
		=\operatorname{Tr}(A\Phi^{-1}\circ\Phi(\rho))\cr
		&=\operatorname{Tr}(\widehat{\Phi^{-1}}(A)
		{\Phi(\rho)})=\langle\widehat{\Phi}^{-1}(A)\rangle_{\Phi(\rho)}.
	\end{align}
	where $\langle\bullet\rangle_\rho=\operatorname{Tr}(\bullet \rho)$ denotes the expectation value of $\bullet$ with respect to the state $\rho$. 
	For the last equality, we have used Eq.~(\ref{eq:AdjInv}). Following the discussion after Eq.~(\ref{eq:AdjVecRep}), $\widehat{\Phi}$ is a CP map. Furthermore, in \cite{Jiaqing2021} it is proven that if a CP map is invertible, then its inverse is Hermiticity preserving. Therefore if $A$ is an observable, so is $\widehat{\Phi}^{-1}(A)$.  Therefore, 
	Eq.~(\ref{eq:mainND}) suggests that by measuring $\widehat{\Phi}^{-1}(A)$ over the noise-affected state, we can recover the noise-free expectation value $\langle A\rangle_{\rho}$. 
	
	This discussion clarifies that in this technique $\Phi^{-1}$ is not implemented physically; it is just used to construct a modified observable  $\widehat{\Phi}^{-1}(A)$. To illustrate that the actual recovery of the expectation value is done in post-processing, we need to recall the definition of the quorum set in quantum tomography.  
	In tomography, a set of measurable observables is called a quorum if the expectation value of any observable can be written in terms of the expectation of the observables belonging to the quorum \cite{DAriano2000}.  Denoting members of the quorum set by $Q_m$, due to the quorum's informational completeness, any observable $A$ can be written as
	
	\begin{equation}
		A=\sum_m a_m Q_m.
	\end{equation}
	Hence, to reconstruct the expectation value of $A$ it is enough to obtain the expectation value of the quorum elements:
	\begin{equation}
		\langle A\rangle_\rho =\sum_{m}a_m \langle Q_m\rangle_\rho.
	\end{equation}
	Different observables correspond to different sets of coefficients $a_m$. For the modified observable $ \widehat{\Phi}^{-1}(A)$ we have
	\begin{align}
		\label{eq:QuorumNoisy}
		\langle \widehat{\Phi}^{-1}(A) \rangle _{\Phi(\rho)}&=\sum_m a_m \langle\widehat{\Phi}^{-1}(Q_m)\rangle _{\Phi(\rho)}\cr
		&=\sum_{m, n} \chi_{n,m} a_m \langle Q_n\rangle  _{\Phi(\rho)},
	\end{align}
	where we have used the fact that $ \widehat{\Phi}^{-1}(Q_m)$ is an observable and hence can be written as a linear combinations of $Q_n$:
	\begin{equation}
		\widehat{\Phi}^{-1}(Q_m)=\sum_n\chi_{n,m} Q_n.
	\end{equation}
	Therefore, for the noise-affected state, the expectation values of $Q_m$ are necessary. The classical post-processing is responsible for computing the summation in Eq.~(\ref{eq:QuorumNoisy}), including the coefficients $\chi_{n,m}$. Hence,  the quantum noise deconvolution technique introduced in \cite{mangini2022, roncallo2023}, is based on classical post-processing and does not require any change in the setup or quantum hardware overhead. 
	\section{Method}
	\label{sec:GenFormalism}
	In this section, we introduce our approach for eliminating the noise effect on the expectation value of an observable when our knowledge of noise is partial. We introduce a measure to quantify the reduction of the noise effect. We discuss the full recovery and partial recovery of the expectation value of observables with our noise deconvolution technique. 
	
	The ultimate aim is to remove the noise effect on the expectation value of an observable $A\in\mathcal{L}(\mathscr{H})$ measured on a quantum state $\rho\in\mathcal{D}(\mathscr{H})$. When there is no noise, the expectation value of the observable is 
	\begin{equation}
		\label{eq:AIdeal}
		\langle A\rangle_{\rm ideal}=\operatorname{Tr}(\rho A)=\dbra{\rho}\dket{A}.
	\end{equation}
	In the presence of noise described by a quantum channel $\Phi$, the expectation value of the observable $A$ is 
	\begin{equation}
		\langle A\rangle_{\rm exp}=\operatorname{Tr}(A\Phi(\rho))=\operatorname{Tr}(\rho\hat{\Phi}(A))=\dbra{\rho}\dket{\hat{\Phi}(A)}.
	\end{equation}
	The deviation of 
	$\langle A\rangle_{\rm exp}$ from its ideal value 
	$\langle A\rangle_{\rm ideal}$,
	represents the error in reading the expectation value of 
	$A$ due to the noise in the experimental setting:
	\begin{align}
		\Delta A_{\rm exp}(\rho)&=\left|\langle A\rangle_{\rm ideal}-\langle A\rangle_{\rm exp}\right|\cr
		&=\left|\dbra{\rho}|\operatorname{id}-\Gamma_\Phi^\dagger\dket{A}\right|.
	\end{align}
	We aim to introduce a technique to remove or reduce the noise effects on the expectation value of a set of observables when we do not have complete knowledge of the noise.  With partial knowledge of noise, we mean that we know the noise belongs to a particular category of quantum channels, but its full characterization is missing. For example, consider a family of channels characterized by some parameters. 
	Knowledge of the family without the precise parameter values constitutes incomplete or partial information.
	 Another example relates to mixed unitary channels. 
	If we know the set of unitary errors but not the probability of having each noise, we say that our knowledge of noise is partial. 
	
	To achieve our goal,  based on the information we have from the noise, we guess the noise and denote it by $\Phi_g$. For example, if we know that the noise belongs to a family of quantum channels characterized by a set of parameters, we guess values for the parameters and construct $\Phi_g$ accordingly. Or if the noise is a mixed unitary channel with known set of unitary Kraus operators and unknown probability for having each, we guess the probability distribution of errors to construct $\Phi_g$. The following reasoning is independent on the quality of the initial guess, $\Phi_g$. If on the noise affected state we measure ${\widehat{\Phi}}^{-1}_g(A)$ and denote the derived expectation value by the noise-deconvolution technique with $\langle A\rangle_{\rm{ND}}$, we have
	\begin{align}
		\label{eq:AND}
		\langle A\rangle_{\rm{ND}}
		&:=\operatorname{Tr}\big({\widehat{\Phi}}^{-1}_g(A)\Phi(\rho)\big)\cr &=\operatorname{Tr}\big( \widehat{\Phi} \circ {\widehat{\Phi}}^{-1}_g(A)\rho\big)\cr 
		&=\operatorname{Tr}\big( \rho\widehat{\Phi}\circ{\widehat{\Phi}}^{-1}_g(A) \big)\cr
		&=\dbra{ \rho}|\Gamma_{\Phi}^\dagger\Gamma^{\dagger}_{\Phi_g^{-1}}\dket{A}.
	\end{align}
	Because our knowledge of noise is partial,  $\langle A\rangle_{\rm{ND}}$ deviates from its original value in the absence of noise $\langle A\rangle_{\rm{ideal}}$.
	The deviation of the expectation value of $A$ after noise deconvolution from its ideal value is denoted by
	\begin{align}
		\label{eq:measure}
		\Delta A_{\rm{ND}}(\rho)& :=  \left|\langle A\rangle_{\rm{ideal}} - \langle A\rangle_{\rm{ND}}\right|\cr
		&=\left|\dbra{\rho}| \mathcal{F}_{\Phi,\Phi_g}\dket{A}\right|,
	\end{align}
	in the last line we have used Eq.~(\ref{eq:AIdeal}) and 
	Eq.~(\ref{eq:AND}) and $\mathcal{F}_{\Phi,\Phi_g}$ is defined as follows
	\begin{equation}
		\label{eq:F}   \mathcal{F}_{\Phi,\Phi_g}:=\operatorname{id} - \Gamma^{\dagger}_\Phi \Gamma^{\dagger}_{\Phi_{g}^{-1}},
	\end{equation}
	If for a density operator $\rho\in\mathcal{D}(\mathscr{H})$ and an observable $A$ we have
	\begin{equation}
		\label{eq:improve}
		\Delta A_{\rm{ND}} (\rho)<\Delta A_{\rm{exp}}(\rho),
	\end{equation}
	then, by the noise deconvolution, the value of the expectation value becomes closer to its ideal value $\langle A\rangle_{\rm{ideal}} $.  
	When inequality in Eq.~(\ref{eq:improve}) holds, the noise deconvolution technique is successful and  $\Delta A_{\rm{ND}}(\rho)$ in Eq.~(\ref{eq:measure}) quantifies how much we can gain by following the procedure. Furthermore, if we have $\Delta A_{\rm ND}(\rho)=0$,  the gain is total and the target expectation value will be obtained by measuring $\widehat{\Phi}^{-1}_g(A)$ on noisy state $\Phi(\rho)$.  
	To make the formalism input-independent, we demand  $\Delta A_{\rm{ND}}(\rho)=0$ for all $\rho\in\mathcal{D}(\mathscr{H})$. In Eq.~(\ref{eq:measure}) if we require $\left|\dbra{\rho}| \mathcal{F}_{\Phi,\Phi_g}\dket{A}\right|=0$,  for all $\dket\rho\in\mathscr{H}_{d^2}$, it results that 
	\begin{equation}
		\dket{A}\in\operatorname{Ker}(\mathcal{F}_{\Phi,\Phi_g}).
	\end{equation}
	Therefore, to obtain the set of observables for which full recovery of the expectation value is possible by measuring $\widehat{\Phi}^{-1}_g(A)$ on $\Phi(\rho)$, we need to construct a general member of $\operatorname{Ker}(\mathcal{F}_{\Phi,\Phi_g})$. 
	After matricizing the solution, we impose the Hermitian constraint to obtain an acceptable set of observables with completely correctable expectation value. 
	\begin{remark}
		When our knowledge of the noise is complete, then the natural guess for the noise is itself: $\Phi_g=\Phi$. In this case, we have 
		\begin{equation}
			\langle A\rangle_{\rm ND}
			=\operatorname{Tr}
			(\rho\widehat{\Phi}\circ\widehat{\Phi}^{-1}(A))=\operatorname{Tr}(\rho A)=\langle A\rangle_{\rm ideal}.
		\end{equation}
		Therefore, by measuring the observable ${\widehat{\Phi}^{-1}}(A)$ after the noise action, one can read $\langle A\rangle_{\rm ideal}$. This special case is completely discussed in \cite{mangini2022} and \cite{roncallo2023} for single and multi-qubit Pauli channels.
	\end{remark}
	\begin{remark}
		If $\operatorname{Ker}(\mathcal{F}_{\Phi,\Phi_g})=\emptyset$ or after imposing the Hermitian condition on the matricized form of a general element of $\operatorname{Ker}(\mathcal{F}_{\Phi,\Phi_g})$ gives zero observable, one can still partially remove the noise effect if by proper guess of noise inequality in Eq.~(\ref{eq:improve}) is satisfied. 
	\end{remark}
	\section{Complete recovery}
	\label{sec:ComRec}
	In this section, we examine our approach for quantum noise deconvolution when knowledge of noise is partial. We address the features of our approach for a generic noise model. Then we focus on mixed unitary channels. In both cases, we provide examples for clarification of the main concept. 
	\subsection{General Noise model}
	In this subsection we discuss our method for a general noise model. After proving a theorem about equivalent channels, we provide two examples for further clarification of the topic. 
	\begin{theorem}
		\label{theorem:Equi}
		Assume that channel $\mathcal{E}$ is equivalent to channel $\Phi$:
		\begin{equation}
			\label{eq:equi2}
			\mathcal{E}=\mathcal{V}\circ\Phi\circ\mathcal{U},
		\end{equation}
		with $\mathcal{U}$ and $\mathcal{V}$ being unitary conjugation with respect to the unitary operators $U$ and $V$. If $\dket{A}\in\operatorname{Ker}(\mathcal{F}_{\Phi,\Phi_g})$ then
		\begin{equation}
			\Gamma_{\widehat{\mathcal{U}}}\dket{A}\in\operatorname{Ker}(\mathcal{F}_{\mathcal{E},\mathcal{E}_g}),
		\end{equation} 
		with
		\begin{equation}
			\mathcal{E}_g={\mathcal{V}}\circ\Phi_g\circ{\mathcal{U}}.
		\end{equation}
	\end{theorem}
	\begin{proof}
		By assumption $\dket{A}\in\operatorname{Ker}(\mathcal{F}_{\Phi,\Phi_g})$. Hence,
		\begin{equation}
			\label{eq:AinKer}
			\Gamma_\Phi^\dagger\Gamma_{\Phi_g^{-1}}^\dagger\dket{A}=\dket{A}.
		\end{equation} 
		Due to the equivalence property in Eq.~(\ref{eq:equi2}) we have
		\begin{equation}
			\label{eq:eqiGamma}
			\Gamma_{\mathcal{E}}=\Gamma_\mathcal{V}\Gamma_\Phi\Gamma_\mathcal{U}.
		\end{equation}
		By combining Eq.~(\ref{eq:AinKer}) and Eq.~(\ref{eq:eqiGamma}) and considering that $\Gamma_{\widehat{\mathcal{U}}}\Gamma_{{\mathcal{U}}}=\Gamma_{\mathcal{U}}\Gamma_{\widehat{\mathcal{U}}}=\operatorname{id}$ we have
		\begin{equation}
			\label{eq:InterStep}
			\Gamma^\dagger_\mathcal{E}\Gamma_\mathcal{V}\Gamma^{\dagger}_{\Phi_g^{-1}}\Gamma_{\mathcal{U}}\dket{B}=\dket{B},
		\end{equation}
		with
		\begin{equation}
			\label{eq:B}
			\dket{B}:=\Gamma_{\widehat{\mathcal{U}}}\dket{A},\;\;{\rm{or}}\;\;B=U^\dagger A U.
		\end{equation}
		Equation~(\ref{eq:InterStep}) implies that $\dket{B}$ belongs to the kernel of $\mathcal{F}_{\mathcal{E},\mathcal{E}_g}$ where 
		\begin{equation}
			\Gamma_{\mathcal{E}_g}=\Gamma_\mathcal{V}\Gamma_{\Phi_g}\Gamma_{\mathcal{U}}.
		\end{equation}
		Or equivalently
		\begin{equation}
			\label{eq:Eg}
			\mathcal{E}_g=\mathcal{V}\circ\Phi_g\circ\mathcal{U}.
		\end{equation}
	\end{proof}
	Theorem~\ref{theorem:Equi}  implies that
	when a family of correctable observables $A$ is derived for a noise model $\Phi$, 
	for all equivalent channels $\mathcal{E}$ as in Eq.~(\ref{eq:equi2}) with guess channels as in Eq.~(\ref{eq:Eg}),  the family of correctable observables is just a unitary conjugation of $A$: $U^\dagger A U$ (see Eq.~(\ref{eq:B})). Hence, the number of parameters of correctable observables for all equivalent channels as in Eq.~(\ref{eq:equi2}) with equivalent guess channels as in
	Eq.~(\ref{eq:Eg}) are the same. 
	
	To illustrate the power of our method, we proceed with an example of noise on a qutrit system. We know the noise belongs to a family of channels with one unknown parameter. 
	\begin{example}\label{exp:qutrit}
		
		In this example we consider a qutrit channel,
		which is not a mixed unitary channel but unital. It is also known as one of the extreme points of the set of qutrit channels:
		\begin{equation}
			\label{eq:PhiW}
			\Phi(\rho)=\sum_{k=1}^3A_k\rho A_k^\dagger,
		\end{equation}
		with Kraus operators:
		\begin{align}
			&A_1=\frac{1}{\sqrt{2}}
			\begin{pmatrix}
				0&0&0\\
				0&1&0\\
				0&0&-1
			\end{pmatrix},\; A_2=\frac{1}{\sqrt{2}}\begin{pmatrix}
				0&-1&0\\
				0&0&0\\
				\omega&0&0
			\end{pmatrix}\cr
			&A_3=\frac{1}{\sqrt{2}}\begin{pmatrix}
				0&0&1\\
				-\omega&0&0\\
				0&0&0
			\end{pmatrix},
		\end{align}
		where $\omega=e^{i\phi}$. It is known that the parameter $\phi\in[0,2\pi)$ but we do not know the exact value of $\phi$. Hence, our knowledge of noise is not complete.  We assume that the value of $\phi=0$. Therefore, the guess channel is described by
		\begin{equation}
			\Phi_g(\rho)=\sum_{k=1}^3 G_k\rho G_k^\dagger,
		\end{equation}
		with Kraus operators $G_k$s:
		\begin{align}
			&G_1=\frac{1}{\sqrt{2}}
			\begin{pmatrix}
				0&0&0\\
				0&1&0\\
				0&0&-1
			\end{pmatrix},\; G_2=\frac{1}{\sqrt{2}}\begin{pmatrix}
				0&-1&0\\
				0&0&0\\
				1&0&0
			\end{pmatrix}\cr
			&G_3=\frac{1}{\sqrt{2}}\begin{pmatrix}
				0&0&1\\
				-1&0&0\\
				0&0&0
			\end{pmatrix}.
		\end{align}
		Following the general formalism introduced in \S\ref{sec:GenFormalism}, by matrix representation of $\Phi$ and $\Phi_g$ we construct $\mathcal{F}_{\Phi,\Phi_g}$:
		\begin{equation}
			\mathcal{F}_{\Phi,\Phi_g}=\operatorname{diag}(0,1-\bar{\omega},1-\bar{\omega},1-\omega,0,0,1-\omega,0,0).
		\end{equation}
		Therefore the general element of the Kernel of 	$\mathcal{F}_{\Phi,\Phi_g}$ is given by
		\begin{equation}
			\label{eq:dektAExample}
			\dket{A}=\begin{pmatrix}
				a&0&0&0&b&c&0&d&e
			\end{pmatrix}^\top.
		\end{equation}
		By matricizing $\dket{A}$ and imposing the Hermitian constraint we have
		\begin{equation}
			\label{eq:CorrObsW}
			A=\begin{pmatrix}
				a&0&0\\
				0&b&c\\
				0&\bar{c}&e
			\end{pmatrix}, a,b,e\in\mathbb{R},\;c\in\mathbb{C}.
		\end{equation}
		Therefore, for a five-parameter family of observables of the form in Eq.~(\ref{eq:CorrObsW}), it is possible to remove the noise effect on the expectation value of these observables by measuring the modifed observable $\widehat{\Phi}^{-1}_g(A)$. For obtaining the modified observable, we note that
		\begin{equation}
			\Gamma_{\Phi_g}=\frac{1}{\sqrt{2}}\begin{pmatrix}
				0&-G_2&G_3\\
				-G_3&G_1&0\\
				G_2&0&-G_1
			\end{pmatrix}
		\end{equation}
		With solving a set of linear equations, the inverse of $\Phi_g$ is found to be
		\begin{equation}
			\label{eq:InvPhig}
			\Gamma_{\Phi_g^{-1}}=\begin{pmatrix}
				\mathcal{X}&\mathcal{P}&\mathcal{Q}\\
				\mathcal{P}^\top&\mathcal{Y}&\mathcal{R}\\
				\mathcal{Q}^\top&\mathcal{R}^\top&\mathcal{Z}
			\end{pmatrix}
		\end{equation}
		with
		\begin{align}
			&\mathcal{X}=\operatorname{diag}(-1,0,0),\\
			&\mathcal{Y}=\operatorname{diag}(0,1,-2),\\
			&\mathcal{Z}=\operatorname{diag}(0,-2,1),
		\end{align}
		and
	   \begin{equation}
	   	\mathcal{P}=\begin{pmatrix}
	   		0&1&0\\
	   		0&0&0\\
	   		-2&0&0
	   	\end{pmatrix},\;
	   	\mathcal{Q}=\begin{pmatrix}
	   	0&0&1\\
	   	-2&0&0\\
	   	0&0&0
	   	\end{pmatrix},\;
	   	\mathcal{R}=\begin{pmatrix}
	   	0&0&0\\
	   	0&0&-1\\
	   	0&0&0
	   	\end{pmatrix},\;
	   \end{equation}
		By using Eq.~(\ref{eq:InvPhig}), the action of $\Gamma^\dagger_{\Phi_g^{-1}}$ on $\dket{A}$ in Eq.~(\ref{eq:dektAExample}) is found. After matricizing that we have
		\begin{equation}
			\label{eq:PhiInvAdjExample}
			\widehat{\Phi}^{-1}_g(A)=\begin{pmatrix}
				-a+b+e&0&0\\
				0&a+b-e&-2c\\
				0&-2\bar{c}&a-b+e
			\end{pmatrix}.
		\end{equation}
			
		For qutrit systems, the whole set of observables has nine real parameters. By this method, the error can be removed from the expectation value of a subfamily of five parameters. Regarding Theorem~\ref{theorem:Equi}, the same can be done for all noise models which are unitarily equivalent to the channel in Eq.~(\ref{eq:PhiW}).
	\end{example}
	Another interesting case is when the noise is a convex combination of known channels, but the weight of each channel in the convex combination is not known. The following is an example of this case:
	\begin{example}
		Here, we consider a two-qubit bit-flip channel with partial memory:
		\begin{equation}
			\label{eq:BF}
			\Phi_{\rm bf}(\bullet)=(1-\mu)\Phi_{\rm bf}^{\rm un}(\bullet)+\mu\Phi^{\rm cc}_{\rm bf}(\bullet),
		\end{equation}
		with $\mu\in[0,1]$ and the uncorrelated and completely correlated channels given by
		\begin{align}    
			&\Phi^{\rm{uc}}_{\rm{bf}}(\rho):= \sum_{i,j=0}^{1} p_i p_j (\sigma_i \otimes \sigma_j)\rho (\sigma_i \otimes \sigma_j), \cr
			&\Phi^{\rm{cc}}_{\rm{bf}}(\rho) := \sum_{i=0}^{1} p_i (\sigma_i \otimes \sigma_i)\rho (\sigma_i \otimes \sigma_i).
		\end{align}
		Here, $p_0=1-p$ and $p_1=p$, with $p\in[0,1]$ denoting the probability of having a bit-flip error. $\sigma_0=\mathbb{I}_2$ is two dimensional identity matrix and $\sigma_1=\sigma_x=\begin{pmatrix}
			0&1\cr
			1 &0
		\end{pmatrix}$ is the usual Pauli matrix. 
		We assume that the value of $p$ is known to us, hence we have complete knowledge about
		$\Phi^{\rm{uc}}_{\rm{bf}}$ and
		$\Phi^{\rm{cc}}_{\rm{bf}}$ appearing in Eq.~(\ref{eq:BF}). But the actual value of $\mu$ is not known to us. 
		
		The superoperator  representations of the two-qubit bit flip channels with partial memory are given by
		\begin{align}
			\Gamma_{\Phi_{\rm bf}^{\rm uc}}&=p(1-p)((\mathbb{I}_2\otimes\sigma_1)^{\otimes 2}+(\sigma_1\otimes\mathbb{I}_2)^{\otimes 2})\cr
			&+(1-p)^2\mathbb{I}^{\otimes 4}+p^2\sigma_1^{\otimes 4},\cr
			\Gamma_{\Phi_{\rm bf}^{\rm cc}}&=(1-p)\mathbb{I}^{\otimes 4}+p\sigma_1^{\otimes 4}.
		\end{align}
		If we guess that $\mu=1$, then 
		\begin{equation}
			\label{eq:PhigBF}
			\Phi_g=\Phi^{\rm{cc}}_{\rm{bf}}.
		\end{equation}
		The superoperator representation of $\Phi^{\rm{cc}}_{\rm{bf}}$ is invertible for $p\neq\frac{1}{2}$:
		\begin{equation}
			\label{eq:InvBFCC}
			\Gamma^{-1}_{\Phi^{\rm cc}_{\rm bf}}=\frac{1-p}{1-2p}\mathbb{I}^{\otimes 4}-\frac{p}{1-2p}\sigma_1^{\otimes 4}.
		\end{equation}
		By using Eqs.~(\ref{eq:F}) and (\ref{eq:InvBFCC}) and considering that 
		\begin{equation}
			\Gamma_{\Phi_{\rm bf}}=(1-\mu)\Gamma_{\Phi^{\rm{uc}}_{\rm bf}}+\mu \Gamma_{\Phi^{\rm{cc}}_{\rm bf}},
		\end{equation}
		we have
		\begin{align}
			\mathcal{F}_{\Phi_{\rm bf}, \Phi^{\rm cc}_{\rm bf}}&=(1-\mu)p(1-p)\cr
			&\times \left(\mathbb{I}_{16}-
			(\mathbb{I}_2\otimes\sigma_1)^{\otimes 2}-(\sigma_1\otimes\mathbb{I}_2)^{\otimes 2}+\sigma_1^{\otimes 4}\right),\cr
		\end{align}
		where $\mathbb{I}_{16}$ is the identity matrix with dimension 16. 
		The Kernel of 	$\mathcal{F}_{\Phi_{\rm bf}, \Phi^{\rm cc}_{\rm bf}}$ is found to be
		\begin{align}
			\operatorname{Ker}\left(\mathcal{F}_{\Phi_{\rm bf}, \Phi^{\rm cc}_{\rm bf}}\right)=\operatorname{Span}\{\ket{++++},\ket{----},&\cr
			\ket{+++-},\ket{++-+},\ket{+-++},\ket{-+++},&\cr
			\ket{---+},\ket{--+-},\ket{-+--},\ket{+---},&\cr
			,\ket{+-+-},\ket{-+-+}\},
		\end{align}
		where $\ket{+}$ and $\ket{-}$ are respectively eigenstates of $\sigma_x$ with eigenvalues $+1$ and $-1$. By matricizing a general element of the $	\operatorname{Ker}\left(\mathcal{F}_{\Phi_{\rm bf}, \Phi^{\rm cc}_{\rm bf}}\right)$ and imposing the Hermiticity constraint, we find a 12 parameter family of observables with correctable expectation values
		\begin{align}
			\label{eq:BFACor}
			A&=\sum_{i=0}^1\sum_{j=2}^3 \left(a_{i,j}\sigma_i\otimes\sigma_j+b_{i,j}\sigma_j\otimes\sigma_i\right)\cr
			&+\sum_{i,j=0}^1 c_{i,j}\sigma_i\otimes\sigma_j.
		\end{align}
		where $a_{i,j}, b_{i,j}, c_{i,j} \in\mathbb{R}$, $\sigma_2=\begin{pmatrix}
			0&-i\\
			i&0
		\end{pmatrix},\;\sigma_3=\begin{pmatrix}
			1&0\\
			0&-1
		\end{pmatrix}$. A general two qubit observable has 16 real parameter. Our proposed noise deconvolution technique introduces a subset of 12-parameter observables with correctable expectation value. At the noise affected system, it is enough to measure
		
		\begin{align}
			\hat{\Phi}^{-1}_g(A)&=\frac{1}{(1-2p)}\sum_{i=0}^1\sum_{j=2}^3 \left(a_{i,j}\sigma_i\otimes\sigma_j+b_{i,j}\sigma_j\otimes\sigma_i\right)\cr
			&+\sum_{i,j=0}^1 c_{i,j}\sigma_i\otimes\sigma_j,
		\end{align}
	
		and make the correction by classical post-processing.
	\end{example}

	\subsection{Mixed Unitary channels}
	In this section, we present how the general formalism we introduced in \S{\ref{sec:GenFormalism} is applied to a mixed unitary channel. First, we clarify what we mean by the partial knowledge of the channel. Then we consider a convex combination of two mixed unitary channels acting on $\mathcal{L}(\mathscr{H}_d)$. We continue with a convex combination of unitary channels on two qubits.
		
		A mixed unitary channel is given by
		\begin{equation}
			\label{eq:RUchannel}
			\Phi(\rho)=\sum_{k}p_kU_k\rho U_k^\dagger,\;\; U_k^\dagger U_k=\mathbb{I}_d \;\forall k, 
		\end{equation}
		where $p_k>0, \forall k$ and $\sum_{k}p_k=1$. We assume that the set of unitary operators $\{U_k\}$ is known, but the probability distribution $\{p_k\}$ is unknown. Hence, our knowledge of the noise is partial. 
		\begin{theorem}
			\label{theorem:RU}
			For a mixed unitary channel as in Eq.~(\ref{eq:RUchannel}), when the probability distribution $\{p_k\}$ is unknown,  observables $A$ with correctable expectation value satisfy the following constraint:
			\begin{equation}
				\label{eq:AB}
				U_kA U_k^\dagger =U_l A U_l^\dagger ,\;\; \forall k , l. 
			\end{equation}
		\end{theorem}
		\begin{proof}
			As the set of Kraus operators contributing to the noise is known, we choose one of them as our guess, e.g, unitary conjugation by $U_g$: 
			\begin{equation}
				\label{eq:GuessUnitary}
				\Phi_g(\bullet)=U_g\bullet U_g^\dagger. 
			\end{equation}
			The inverse and the adjoint of a unitary conjugation are the same: $\mathcal{E}^{-1}_g=\widehat{\mathcal{E}_g}$ or equivalently for the  superoperator representation of the guess channel we have,
			\begin{equation}  
				\label{eq:GammaGuessUnitary}
				\Gamma^\dagger_{\mathcal{E}^{-1}_g}=\Gamma_{\mathcal{E}_g}.
			\end{equation}
			Therefore, for $\mathcal{F}_{\Phi,\Phi_g}$ defined in Eq.~(\ref{eq:F}) we have
			\begin{equation}
				\mathcal{F}_{\Phi,\Phi_g}=\sum_{i\neq g} p_i (\operatorname{id}-\Gamma_i),
			\end{equation}
			with
			\begin{equation}
				\label{eq:Gi}
				\Gamma_i=U_i^\dagger U_g\otimes U_i^\top U_g^*.
			\end{equation}
			Because we do not know the probability of each error $U_i$, to find the observable with correctable expectation value we require $\mathcal{F}_{\Phi,\Phi_g}\dket{A}=0$ for  any vector of probability distribution $\{p_i\}_i$ which is equivalent to solving the following:
			\begin{equation}
				\label{eq:Stab|Gammai}
				\Gamma_i\dket{A}=\dket{A},\;\;\forall i.
			\end{equation}
			When $i=g$,  $\Gamma_g=\operatorname{id}\otimes\operatorname{id}$ and the above equality is trivial. From equation (\ref{eq:Stab|Gammai}), which states that $\dket{A}$ is the invariant state under all  $\Gamma_i$s and definition of $\Gamma_i$ in Eq.~(\ref{eq:Gi})  we conclude that
			\begin{equation}
				\label{eq:RU-C}
				U_g AU_g^\dagger =U_i A U_i^\dagger,\;\;\forall i,
			\end{equation}
		     which is equivalent to the conclusion of the Theorem as stated in Eq.~(\ref{eq:AB}).
		     		\end{proof}
		     		Theorem~\ref{theorem:RU} gives the constraints that are satisfied by an observable $A$ with correctable expectation value, when noise is a mixed unitary channel and the probability of each error is unknown. Furthermore, for such noise models, Theorem~\ref{theorem:RU} gives the instruction for constructing the set of observable with correctable expectation values. To find the desired set of observables one should find $\mathcal{S}_i$, the set of invariant vectors under the action of $\Gamma_i$ 
			\begin{equation}
				\mathcal{S}_i=\{\dket{v}\in\mathscr{H}_{d^2} \;\;| \;\;\Gamma_i\dket{v}=\dket{v}\}.
			\end{equation}
			The intersection of these sets determines $\dket{A}$:
			\begin{equation}
				\dket{A}\in\bigcap_{i}\mathcal{S}_i.
			\end{equation}
			After obtaining $\dket{A}$, it must be matricized, and then the Hermitian constraint $A=A^\dagger$ must be imposed. This may reduce the number of free parameters in $A$. 
			By measuring 
			\begin{equation}
				\hat{\Phi}_g^{-1}(A)=U_g^\dagger  AU_g,
			\end{equation}
			we can completely recover the expectation value of $A$ for all possible initial states. 
		From Eq.~(\ref{eq:RU-C}), it is clear that the set of observables with correctable expectation value is independent of our choice for $U_g$. For any choice of $U_g$, observable $\hat{\Phi}_g^{-1}(A)=U_g^\dagger  AU_g$,  must be measured on $\Phi(\rho)$, and perform classical post-processing to have $\Delta A_{\rm ND}(\rho)=0$ for all density operators $\rho$. 
		
		The following is an example of a mixed unitary channel for modeling noise. 
		\begin{example}
			Consider a mixed unitary channel of the form in Eq.~(\ref{eq:RUchannel}) with
			
			\begin{align}
				&U_1=\frac{1}{\sqrt{12}}\begin{pmatrix}
					\sqrt{2}+1&  \sqrt{6} &\sqrt{2}-1\\
					\sqrt{2}+1 & -\sqrt{6} & \sqrt{2}-1\\
					\sqrt{2}-2 & 0 & \sqrt{2}+2
				\end{pmatrix},\\
				&U_2=\frac{1}{\sqrt{12}}\begin{pmatrix}
					\sqrt{2}+1&  i\sqrt{6} &\sqrt{2}-1\\
					\sqrt{2}+1 & -i\sqrt{6} & \sqrt{2}-1\\
					\sqrt{2}-2 & 0 & \sqrt{2}+2
				\end{pmatrix},\\
				&U_3=\frac{1}{\sqrt{12}}\begin{pmatrix}
					\sqrt{2}+i&  i\sqrt{6} &\sqrt{2}-i\\
					\sqrt{2}+i & -i\sqrt{6} & \sqrt{2}-i\\
					\sqrt{2}-2i & 0 & \sqrt{2}+2i
				\end{pmatrix}.
			\end{align}
			As discussed, our knowledge of noise is partial, and we do not know the probability of having different $U_i$s as errors. Without loss of generality, we take $U_g=U_1$. The subspace which is invariant  under $\Gamma_2$ (defined in Eq.~(\ref{eq:Gi})) is given by
			\begin{equation}
				\mathcal{S}_2=\operatorname{Span}\{\ket{\lambda_1,\lambda_1}, \ket{\lambda_2,\lambda_2}, \ket{\lambda_3,\lambda_3}, \ket{\lambda_1,\lambda_3}, \ket{\lambda_3,\lambda_1}\},
			\end{equation}
			where $\ket{\lambda_1},\ket{\lambda_2}$ and $\ket{\lambda_3}$ are respectively eigenvectores of $U_2^\dagger U_1$ with eigenvalues $1,-i,1$:
			\begin{equation}
				\ket{\lambda_1}=\begin{pmatrix}
					1\\0\\0
				\end{pmatrix},\;\ket{\lambda_2}=\begin{pmatrix}
					0\\1\\0
				\end{pmatrix},\;
				\ket{\lambda_3}=\begin{pmatrix}
					0\\0\\1
				\end{pmatrix}.
			\end{equation}
			The subspace invariant under $\Gamma_3$ is given by
			\begin{equation}
				\mathcal{S}_3=\operatorname{Span}\{\ket{\mu_1,\mu_1}, \ket{\mu_2,\mu_2}, \ket{\mu_3,\mu_3}, \ket{\mu_1, \mu_3}, \ket{\mu_3,\mu_1}\},
			\end{equation}
			where $\ket{\mu_1}, \ket{\mu_2}$ and $\ket{\mu_3}$ are respectively eigenvector of $U_3^\dagger U_1$ with eigenvalues $-i,1,-i$:
			\begin{equation}
				\ket{\mu_1}=\frac{1}{\sqrt{2}}\begin{pmatrix}
					1\\0\\-1
				\end{pmatrix},\;\ket{\mu_2}=\frac{1}{\sqrt{2}}\begin{pmatrix}
					1\\0\\1
				\end{pmatrix},\;
				\ket{\mu_3}=\begin{pmatrix}
					0\\1\\0
				\end{pmatrix}.
			\end{equation}
			By finding the intersection of $\mathcal{S}_2$ and $\mathcal{S}_3$, matricizing a general element of the intersection and then imposing the Hermitian constraint, we get that observables with correctable expectation values belong to the set of three-parameter observables of the form:
			\begin{equation}
				\begin{pmatrix}
					a &0 & b\\
					0 & c & 0\\
					b & 0 & a
				\end{pmatrix},
			\end{equation}
			where $a,b,c\in\mathbb{R}.$
		\end{example}
		As discussed in Theorem~\ref{theorem:RU}, to obtain observables 
		for which $\Delta A_{\rm ND}(\rho)=0,\;\forall\rho\in\mathcal{D}(\mathscr{H})$ one has to construct 
		the sets $\mathcal{S}_i$ and their intersection. If $\cap_i\mathcal{S}_i=\emptyset$, or if the matricized form of elements in $\cap_i\mathcal{S}_i$ do not admit Hermitian matrix form, then by this approach we can not remove the noise effect on the expectation value of any observable. The following lemma, introduces a subset of mixed unitary channels, for which the instruction of Theorem~\ref{theorem:RU}  does not provide any non-trivial observable with correctable expectation value. 
		\begin{lemma}\label{lemma:IrepU}
			For noise models in Eq.~(\ref{eq:RUchannel}), if $U_i$ are unitary irreducible representations of some groups, then the only observable for which $\Delta A_{\rm ND}(\rho)=0$ for all $\rho\in\mathcal{D}(\mathscr{H})$,  is proportional the identity operator. 
		\end{lemma} 
		\begin{proof}
			By assumption, the error operators $U_i$ are representations of a group. Due to the group structure of the error operators $U_i^\dagger U_g=U_k$ is another element of the set of errors and condition in Eq.~(\ref{eq:RU-C}) becomes
			\begin{equation}
				[U_k, A]=0, \;\;\forall k.
			\end{equation}
			Because it is assumed that $U_i$s are irreps of the group, due to Schur lemma, $A$ is proportional to identity.  
		\end{proof}
		This lemma demonstrates that, for a general qubit Pauli channel as defined in Eq.~(\ref{eq:RUchannel}), the identity is the only observable with a correctable expectation value—a conclusion reached without requiring detailed calculation.
		It is just enough to notice that the set $\{\pm I,\pm i\sigma_1,\pm i\sigma_2,\pm i\sigma_3\}$ is a two-dimensional irreducible representation of the Pauli group or Quaternion group $Q_8$.  Regarding Theorem~\ref{theorem:Equi} and considering that all unital qubit channels are equivalent to a Pauli channel \cite{ruskai2002}, we conclude that for all qubit unital channels, by the approach proposed in Theorem~\ref{theorem:RU} we can not correct the expectation value of any non-trivial observable.
		
		Lemma~\ref{lemma:IrepU} gives a subset of channels for which the strategy in Theorem~\ref{theorem:RU} is not useful. On the contrary, the following theorem applies to the noise model where the strategy of Theorem~\ref{theorem:RU} enables us to recover the expectation value of a family of observables. The number of parameters of this family of observables is addressed in the following theorem.
		\begin{theorem}\label{theorem:UV}
			Consider a noise model
			\begin{equation}
				\label{eq:RU2}
				\Phi(\bullet)=(1-p)U_1\bullet U_1^\dagger+p U_2\bullet U_2^\dagger, \;\; 0\leq p\leq 1,
			\end{equation}
			with non-trivial unitary noise operators $U_1, U_2\in \mathcal{L}(\mathscr{H}_d)$. In the absence of knowledge about the probability $p$, perfect noise deconvolution 
			($\Delta A_{\rm ND}(\rho)=0$), for all $\rho\in\mathcal{D}(\mathscr{H}_d)$) is achievable for any observable $A$ belonging to the $d$-dimensional subspace of observables: 
			 $\operatorname{Span}\{\ket{\omega_k}\bra{\omega_k}\}$, where $\ket{\omega_k}$ are eigenstates of $U_1^\dagger U_2$.
		\end{theorem}		
		\begin{proof}
			Without loss of generality we take 
			\begin{equation}
				\Phi_g(\bullet)=U_2\bullet U_2^\dagger.
			\end{equation}
			According to Theorem~\ref{theorem:RU}, Eq.~(\ref{eq:RU-C}) the observable with completely correctable expectation value satisfies:
			\begin{equation}
				U_1 A U_1^\dagger=U_2 A U_2^\dagger.
			\end{equation}
			To construct $A$, we need to obtain the kernel of $	\mathcal{F}_{\Phi,\Phi_g}$ defined in  Eq.~(\ref{eq:F}):
			\begin{equation}
				\label{eq:F_UV}
				\mathcal{F}_{\Phi,\Phi_g}= (1-p)(\operatorname{id}-W\otimes W^*),
			\end{equation}
			where 
			\begin{equation}
				\label{eq:W}
				W=U_1^\dagger U_2.
			\end{equation}
			Therefore, for vectors $\dket{A}$ in the kernel of $	\mathcal{F}_{\Phi,\Phi_g}$ we have
			\begin{equation}
				W\otimes W^*\dket{A}=\dket{A}.
			\end{equation}
			The operator $W$ is a unitary operator. Hence its eigenvalues are pure phase: 
			\begin{align}
				&W\ket{\omega_k}=e^{i\omega_k}\ket{\omega_k},\cr
				&W^*\ket{\omega_k^*}=e^{-i\omega_k}\ket{\omega_k^*},
			\end{align}
			for $k\in[0,d-1]=\{0,1,\cdots,d-1\}$, $\omega_k\in\mathbb{R}$, and $\ket{\omega_k^*}$ denoting the complex conjugate of $\ket{\omega_k}$.  Therefore,
			\begin{equation}
				(W\otimes W^*)\dket{\omega_k,\omega^*_k}=\dket{\omega_k,\omega^*_k},\;\; \forall k\in[0,d-1],
			\end{equation}
			and the kernel of $\mathcal{F}_{\Phi,\Phi_g}$ is $d$-dimensional:
			\begin{equation}
				\operatorname{Ker}(\mathcal{F}_{\Phi,\Phi_g})=\operatorname{Span}\{\dket{\omega_k,\omega^*_k}\}, \;\;\; k\in[0,d-1].
			\end{equation}
			It proves that the expectation value of $d$-parameter family of observables 
			\begin{equation}
				\label{eq:Observable2Uni}
				A=\sum_{k=0}^{d-1}a_k \ket{\omega_k}\bra{\omega_k}, \;\; a_k\in\mathbb{R},
			\end{equation}
			can be completely recovered by measuring $\widehat{\Phi}^{-1}(A)$ on the noise affected state. Here we have assumed that the spectrum of $W$ is non-degenerate. The degeneracy in the spectrum of $W$ gives a larger family of observables for which $\Delta A_{\rm ND}(\rho)=0$. Hence, the minimum number of parameters of desired observables is $d$. 
		\end{proof}
		
		\begin{remark}
			Channel
			\begin{equation}
				\label{eq:IW}
				\mathcal{E}(\rho)=(1-p)\rho+W \rho W^\dagger,
			\end{equation}
			is equivalent to channel $\Phi$ in Eq.~(\ref{eq:RU2}):
			\begin{equation}
				\mathcal{E}(\bullet)=U_1^\dagger\Phi(\bullet) U_1.
			\end{equation}
			Therefore by applying Theorem~\ref{theorem:Equi} and Theorem~\ref{theorem:UV} we understand that when the noise is described by $\mathcal{E}$ in Eq.~(\ref{eq:IW}), the set of observables with correctable expectation values has $d$-parameter. On the other hand, according to Eq.~(\ref{eq:RU-C}):
			\begin{equation}
				WAW^\dagger=A.
			\end{equation}
			Therefore, $[W^\dagger,A]=0$, from which we conclude that $\widehat{\mathcal{E}_g}^{-1}(A)=\widehat{\mathcal{E}}(A)=A$. This means that the set of observables with correctable expectation value, are invariant under the noise. Furthermore, by construction we must measure $\widehat{\mathcal{E}_g}^{-1}(A)$ which is equal to $A$. The same argument is not valid for the equivalent channel in Eq.~(\ref{eq:RU2}): $\widehat{\Phi}(A)\neq A$ and $\widehat{\Phi}_g^{-1}(A)\neq A$. 
		\end{remark}
		The following is an explicit example of Theorem~\ref{theorem:UV}.
		\begin{example}
			Suppose that the unitary operators $U_1$ and $U_2$ in Eq.~(\ref{eq:RU2}) are
			\begin{equation}
				U_1=\frac{1}{\sqrt{2}}\begin{pmatrix}
					1 & -1\cr
					1 & 1
				\end{pmatrix},\;\;
				U_2=\begin{pmatrix}
					0 & 1\cr
					1 & 0
				\end{pmatrix}.
			\end{equation}
			For guess $\Phi_g(\bullet)=U_2\bullet U_2^\dagger$. Hence,
			\begin{equation}
				W=U_1^\dagger U_2=\frac{1}{\sqrt{2}}\begin{pmatrix}
					1 & 1\cr
					1 & -1
				\end{pmatrix},
			\end{equation}
			with eigenvalues $1$ and $-1$ and respective eigenvectors:
			\begin{equation}
				\ket{\omega_1}=\begin{pmatrix}
					1+\sqrt{2}\cr
					1
				\end{pmatrix},\;\; \ket{\omega_2}=\begin{pmatrix}
					1-\sqrt{2}\cr
					1
				\end{pmatrix}.
			\end{equation}
			According to Eq.~(\ref{eq:Observable2Uni}) the expectation value of the following observables can be completely recovered:
			\begin{equation}
				\label{eq:exm-2Uqubit}
				A=\begin{pmatrix}
					a+2b & b\cr
					b & a
				\end{pmatrix},\;\; a,b\in\mathbb{R}.
			\end{equation}
			An arbitrary qubit observable has 4 parameters. By this method the expectation value of a subset of 2-parameter observables in Eq.~(\ref{eq:exm-2Uqubit}) can be recovered by measuring
			\begin{equation}
				U_2A U_2^\dagger=\begin{pmatrix}
					a & b\cr
					b & a+2b
				\end{pmatrix}, 
			\end{equation}
			and classical post-processing. 
		\end{example}
		As given in the statement of Theorem~\ref{theorem:UV}, for a system with $d$-dimensional Hilbert space, the number of parameters of the observables with completely recoverable expectation value, is at least $d$. If the  spectrum of $W$ in Eq.~(\ref{eq:W}) is degenerate then the number of parameters can be more than $d$. The following example shows this matter more clearly:
		\begin{example}
			Consider the noise model as in Eq.~(\ref{eq:RU2}) with
			\begin{equation}
				U_1=\frac{1}{\sqrt{2}}\begin{pmatrix}
					1 & 1 & 0\\
					1 & -1& 0\\
					0 & 0 & \sqrt{2}
				\end{pmatrix},\;
				U_2=\frac{1}{3\sqrt{2}}\begin{pmatrix}
					4 & 1 & 1\\
					0 & 3& -3\\
					-\sqrt{2} & 2\sqrt{2} & 2\sqrt{2}
				\end{pmatrix}.
			\end{equation}
			Then by taking $U_g=U_2$, the operator $W$ defined in Eq.~(\ref{eq:W}) is given by
			\begin{equation}
				W=\frac{1}{3}\begin{pmatrix}
					2 &2 & -1\\
					2 &-1&2\\
					-1&2 & 2
				\end{pmatrix}.
			\end{equation}
			The spectrum of $W$ is degenerate with eigenvalues $\omega_1=1,\omega_2=-1$ and $\omega_3=1$. The corresponding eigenvector of $W$ are respectively given by
			\begin{equation}
				\ket{\omega_1}=\frac{1}{\sqrt{3}}
				\begin{pmatrix}
					1\\1\\1
				\end{pmatrix},\;
				\ket{\omega_2}=\frac{1}{\sqrt{6}}
				\begin{pmatrix}
					1\\-1\\1
				\end{pmatrix},\;
				\ket{\omega_3}=\frac{1}{\sqrt{2}}
				\begin{pmatrix}
					1\\0\\-1
				\end{pmatrix}.
			\end{equation}
			Therefore, 	
			\begin{align}
				\operatorname{Ker}(\mathcal{F}_{\Phi,\Phi_g})=\operatorname{Span}\{
				&\ket{\omega_1,\omega_1^*},\ket{\omega_2,\omega_2^*},	\ket{\omega_3,\omega_3^*}, \cr
				&\ket{\omega_1,\omega_3^*},\ket{\omega_3,\omega_1^*}\}.
			\end{align}
			After imposing the Hermiticity constraint for the observables, we have a five parameter family of
			observables with correctable expectation value
			\begin{align}
				A&=a\ket{\omega_1}\bra{\omega_1}+b\ket{\omega_2}\bra{\omega_2}+c\ket{\omega_3}\bra{\omega_3}\cr
				&+d\ket{\omega_1}\bra{\omega_3}+d^*\ket{\omega_3}\bra{\omega_1},
			\end{align}
			with $a,b,c\in\mathbb{R}$ and $d\in\mathbb{C}$. Hence, this family is defined by five real parameters. 
		\end{example}
		\section{Partial denoising}
		\label{sec:ParRec}

		In this section, we discuss the possibility of reducing the noise effect on the expectation value of observables when full recovery is not possible. We address the requirements for reducing the noise effect and complete the section with an example. 
		
		As discussed in \S\ref{sec:GenFormalism},
		for a given noise model $\Phi$ with guess channel $\Phi_g$, the kernel of
		$\mathcal{F}_{\Phi,\Phi_g}$  is the key element for constructing  the set of observables with correctable
		expectation values. The discussed technique is input-independent. That is for all $\rho\in\mathcal{D}(\mathscr{H})$, 
		$\Delta A_{\rm ND}(\rho)=0$. 
		For observables not belonging to this family, noise effects on the expectation value can still be reduced if the inequality in Eq.~(\ref{eq:improve}) is satisfied. However, the exact reduction of the noise is input-dependent.
		
		While allowing full knowledge of the input state is not unrealistic, we can restrict our analysis to input states belonging to a particular subset of density operators. For instance, states generated by a process with known properties or inputs prepared with specific characteristics. With such partial prior knowledge about the input, it becomes possible to go beyond the set of observables with completely correctable expectation value. To this end, for a given noise model and guess of the noise, one needs to determine a subset of density operators and a set of observables for which inequality in Eq.~(\ref{eq:improve}) is satisfied. The following example clarifies the matter. 
		
		\begin{example}
			For the set of 12-parameter observables in Eq.~(\ref{eq:BFACor}), $\Delta A_{\rm ND}(\rho)=0,\;\forall\rho\in\mathcal{D}(\mathscr{H})$. That means the effect of noise can be completely removed by the noise deconvolution method. As discussed in the general setting for observables satisfying inequality in Eq.~(\ref{eq:improve}) although the complete correction is not possible, but the value of expectation value after implementing noise deconvolution technique becomes closer to its original value. For example consider the following observable which does not belong to the set of observables in Eq.~(\ref{eq:BFACor}):
			\begin{equation} \label{eq:sampleofotherA}
				B = \sigma_0 \otimes \sum_{i=2}^{3}\sigma_i + \sum_{i=2}^{3}\sigma_i \otimes  \sigma_0 + \sum_{i,j=2}^{3}\sigma_i \otimes  \sigma_j.
			\end{equation} 
			For this observable and a family of two-qubit states of the form
			\begin{equation}
				\label{eq:rhox}
				\rho =\frac{1}{4}\left(\sigma_0 \otimes \sigma_0 + x(\sigma_2 \otimes \sigma_0 + \sigma_0 \otimes \sigma_3) + \sigma_2 \otimes \sigma_3\right),
			\end{equation} 
			where $x$ is an unknown parameter in the interval $0<x\leq 1$, we have
			\begin{equation}
				\label{eq:DBexp}
				\Delta B_{\rm{exp}} (\rho)= p[(1-p)(1-\mu)+x].
			\end{equation}
			This quantifies the deviation of the expectation value of $B$ in Eq.~(\ref{eq:sampleofotherA}) with respect to the noise affected state, from its ideal value. The parameter $x$  reflects the dependence of this quantity on initial state. By considering $\Phi_g$ as in Eq.~(\ref{eq:PhigBF}) and recalling Eq.~(\ref{eq:measure}) we have
			\begin{equation}
				\label{eq:DBND}
				\Delta B_{\rm{ND}} (\rho)= p (1-p)(1-\mu).
			\end{equation}
			By comparing Eq.~(\ref{eq:DBexp}) with Eq.~(\ref{eq:DBND}) we have
			\begin{equation}
				\Delta B_{\rm{ND}}(\rho) \leq \Delta B_{\rm exp}(\rho),
			\end{equation}
			for all $\rho$ in Eq.~(\ref{eq:rhox}). 
			This inequality holds for all value of the unknown parameters $x$ and $\mu$. Hence, for all states in Eq.~(\ref{eq:rhox}) and for all values of $\mu$, there is relative success in recovering the expectation value of observable $B$ in Eq.~(\ref{eq:sampleofotherA}).
	
			 By using Eq.~(\ref{eq:DBexp}) and Eq.~(\ref{eq:DBND}) we see that 
			\begin{equation}
				\Delta B_{\rm{ND}}(\rho)- \Delta B_{\rm exp}(\rho)=-px
			\end{equation}
			Hence, in this example for input states with larger values of $x$, the noise deconvolution technique is more effective. 
		\end{example}
		
		\section{Conclusion}
		\label{sec:Dis}
		In this work, we introduce a new quantum noise deconvolution technique to recover the exact expectation value of a set of observables. Unlike full channel correction, which demands significant quantum resources, our method employs classical post-processing of measurement data. This makes it a practical tool for error mitigation on NISQ-era devices.
		Existing noise deconvolution techniques rely on complete knowledge of noise, which requires resources like process tomography. In \cite{roncallo2023}, it is shown that when complete knowledge of noise is not available, by preparation of a particular initial state, without full process tomography, it is possible to recover the ideal expectation value of the observable. But this method requires state preparation and becomes less efficient when the noise is non-unital. 
		
		Our method does not need complete knowledge of noise, and does not require even partial quantum process tomography. It is applicable to all noise models in arbitrary dimensions and also to multi-partite systems. 
		
		It is shown in \cite{mangini2022} that when noise is described by $\Phi$ to recover the expectation value of $A$, one must measure $\widehat{\Phi}^{-1}(A)$. 
		Here, as the exact channel $\Phi$ is unknown, we work with a hypothesized model, referred to as ``guessed channel"  $\Phi_g$, which is constructed from our partial knowledge of noise.
		Instead of $\widehat{\Phi}^{-1}(A)$, which we cannot have due to our incomplete knowledge of noise, we measure $\widehat{\Phi}^{-1}_g(A)$ to obtain the noise deconvolved expectation value. To benchmark the efficiency of our method, we define a measure $\Delta A_{\rm ND}(\rho)$ that represents the absolute deviation between the deconvolved expectation value of the observable and its ideal value. If $\Delta A_{\rm ND}(\rho)$ is zero or less than the absolute deviation between the noisy expectation value of the observable and its ideal value namely $\Delta A_{\rm exp}(\rho)$, we consider the noise deconvolution successful.  By demanding $\Delta A_{\rm ND}(\rho)=0$ we obtain a set of observables with complete recoverable expectation value if $\widehat{\Phi}^{-1}_g(A)$ is measured on the noisy states. Furthermore, we addressed the partial recovery of the expectation value of observables when full recovery is not possible, at least with the guessed channel.  We discuss that the performance of partial recovery of expectation value is input-dependent and it relies on partial knowledge about the system. 
		
		While our method advances quantum noise deconvolution by eliminating the need for complete knowledge of the noise, it also raises new questions for further improvements in quantum noise deconvolution. For instance, given a noise model, it is interesting to investigate what the optimal guessed channel is that gives the largest set of observables with completely recoverable expectation value. Furthermore, constructing the largest set of observables with partially recoverable expectation value and quantifying how the number of parameters of this set scales with partial knowledge of the input state enriches our quantum noise deconvolution techniques.  
		
		\begin{acknowledgments}
			L.M. acknowledges
			financial support from the Iran National Science Foundation (INSF) under Project No. 4022322, and financial support by Sharif University of Technology, Office
			of Vice President for research, under Grant No. QST4040206, and support from the ICTP through the Associates Programme (2019-2024).  
		\end{acknowledgments}
		\appendix
		\section{Inverse of the adjont map}
		 As discussed in \S\ref{sec:GenFormalism}, when we are interested in the expectation value of an observable $A$, on need to measure the modified observable $\hat{\Phi}_g^{-1)}$ on the noise affected system $\Phi(\rho)$. In this appendix we review the steps for computing the inverse of the adjoint of a map $\Phi$ or $\hat{\Phi}^{-1}$. For an explicit example, we discuss how Eq.~(\ref{eq:PhiInvAdjExample}) has been derived. 
		\bibliography{Main}
	\end{document}